\documentclass[letterpaper, 10 pt, conference]{ieeeconf}

\IEEEoverridecommandlockouts

\usepackage{amsmath}
\usepackage{color}
\usepackage{algorithm}
\usepackage{tikz}
\usepackage{graphicx}
\usepackage{authblk}
\usepackage{mathtools}
\usepackage{tikz}
\usetikzlibrary{arrows,automata}
\usepackage{amsfonts}
\usepackage[latin9]{inputenc}

\title{\LARGE \bf  Robust Temporal Logic Model Predictive Control}

\author{Sadra Sadraddini and Calin Belta 
\thanks{The authors are with the Department of Mechanical Engineering, Boston University, Boston, MA 02215 \{sadra,cbelta\}@bu.edu \\ This work was partially supported by the NSF under grants CPS-1446151 and CMMI-1400167.
\\
This work has been accepted to appear in the proceedings of 53rd Annual Allerton Conference on Communication, Control and Computing, Urbana-Champaign,
IL (2015)
\\
\textcopyright 2015 IEEE. Personal use of this material is permitted. Permission from IEEE must be obtained for all other uses, in any current or future media, including reprinting/republishing this material for advertising or promotional purposes, creating new collective works, for resale or redistribution to servers or lists, or reuse of any copyrighted component of this work in other works. 
}
}

\newtheorem{define}{Definition}
\newtheorem{example}{Example}
\newtheorem{problem}{Problem}
\newtheorem{lemma}{Lemma}
\newtheorem{remark}{Remark}

\newtheorem{theorem}{Theorem}
\newtheorem{corollary}{Corollary}
\newtheorem{proposition}{Proposition}

\makeatletter
\def\ps@IEEEtitlepagestyle{%
  \def\@oddfoot{\mycopyrightnotice}%
  \def\@evenfoot{}%
}
\def\mycopyrightnotice{%
  {\footnotesize The copyright belongs to me!\hfill}
  \gdef\mycopyrightnotice{}
}

\usepackage{amsfonts}

\begin{document}
\maketitle

 \thispagestyle{empty}
\pagestyle{empty}
\makeatother

\begin{abstract}
Control synthesis from temporal logic specifications has gained popularity in recent years. 
In this paper, we use a model predictive approach to control discrete time linear systems with additive bounded disturbances subject to constraints given as formulas of signal temporal logic (STL). 
We introduce a (conservative) computationally efficient framework to synthesize control strategies based on mixed integer programs. The designed controllers satisfy the temporal logic requirements, are robust  to all possible realizations of the disturbances, and optimal with respect to a cost function. In case the temporal logic constraint is infeasible, the controller satisfies  a relaxed, minimally violating constraint.  An illustrative case study is included. 
\end{abstract}

\section{Introduction}
Model predictive control (MPC), also known as receding horizon control (RHC), is a popular approach to generate
(sub)optimal control strategies for systems with constraints \cite{garcia1989model},\cite{camacho2013model}. During the last decades, many theoretical aspects of MPC, such as stability and robustness, have been investigated \cite{Kothare1996},\cite{Mayne2000}. However, most works in the MPC literature focus on simple classes of performance objectives and constraints such as closeness to a reference point or trajectory. Recently, there has been a growing trend in control theory in considering a richer class of constraints that are described using rules and symbolism from formal methods such as temporal logics. 

Temporal logics \cite{Baier2008}, such as linear temporal logic (LTL), computational tree logic (CTL), and signal temporal logic (STL), are able to describe a wide range of specifications. For example, satisfying disjoint sets of constraints infinitely often with specific deadlines for the satisfaction of each of them, (i.e., oscillatory behavior), can be naturally expressed in a temporal logic such as STL. 

Temporal logic control based on finite abstractions \cite{kloetzer2008fully},\cite{kress2009temporal},\cite{tabuada2006linear} is a correct by construction control synthesis method that involves a finite state machine representation of the control system that is generally expensive to compute. To address this limitation, some works proposed receding horizon approaches to temporal logic control  \cite{wongpiromsarn2010receding}, \cite{AydinGol2015}. Inspired by \cite{bemporad1999control}, the authors in \cite{karaman},\cite{wolff2013optimal} have developed methods for translation of LTL constraints into mixed-integer constraints that are used in the controller synthesis algorithm. More recently, \cite{raman}, \cite{raman2015reactive} have extended this methodology to MPC from STL specifications by developing mixed-integer encodings of bounded time model checking.

In this paper, we use STL formulas over predicates in the state of the system to specify correctness requirements. We focus on discrete-time linear systems with additive bounded uncertainties. We propose an MPC approach to the synthesis problem with the goal of satisfying the STL specification globally (i.e., for all times) by characterizing the bounded propagation of uncertainties into STL constraints. We take a conservative approach that is computationally as tractable as STL control of deterministic systems. Furthermore, the notions introduced in this paper enable to treat STL constraints as \emph{soft constraints} that may be violated if a feasible control policy does not exist. We are thus able to find minimally violating solutions in the presence of uncertainties and limited control actuation.

In the closest related work, the authors in \cite{raman2015reactive} use a counter example guided approach to receding horizon control of disturbed systems. A major disadvantage of this approach is that the generation of counter examples may never terminate. Also, taking into account a large number of counter examples is computationally intractable. Furthermore, since there does not exist a global feasibility guarantee for STL MPC, the control algorithm in \cite{raman2015reactive} may encounter infeasibility, an issue that we address in this paper.

This paper is organized as follows. First, we informally state the problem in Section \ref{sec:problem}. Next, the receding horizon mechanism of STL control is explained in Section \ref{sec:rhc}. After providing technical details about robust prediction in Section \ref{sec:robust}, we formalize the problem as an optimization problem in Section \ref{sec:optimization}. Finally, a case study is presented in Section \ref{sec:case}.

\section{Problem Statement}
\label{sec:problem}
We consider discrete-time dynamical systems of the form
\begin{equation}
\label{eq:dynamics}
x[t+1]=Ax[t]+Bu[t]+w[t],
\end{equation}
\begin{equation}
y[t]=Cx[t]+Du[t]+e,
\label{eq:secondary_signals}
\end{equation}
where $x \in \mathbb{R}^n $ is the state, $u \in \mathcal{U} \subseteq \mathbb{R}^m$ is the control restricted to an admissible set
$\mathcal{U}$ and $t \in \mathbb{Z}_{\ge 0}$ is time. The exogenous inputs, or additive disturbances, $w[t] \in \mathcal{W} \subset \mathbb{R}^n$ are unknown but restricted to a known bounded set $\mathcal{W}$, which is assumed to be a polytope. Matrices $A,B$ are fully known with appropriate dimensions. We also consider an associated stage cost function $J(x[t],u[t]), J:\mathbb{R}^n \times \mathcal{U} \rightarrow \mathbb{R}$, defined for each time step. The output $y$, and the corresponding $C$, $D$ and $e$ matrices, are defined from the predicates present in the temporal logic specification, as it will be shown below.

STL is a formal framework for describing a wide range of specifications in a convenient and compact form. The formal definition of STL is not presented in this paper, and the interested reader is referred to \cite{maler_stl,donze}. Informally, STL formulas consist of boolean connectives $\neg$ (negation), $\wedge$ (conjunction), $\vee$ (disjunction), and bounded-time temporal operators  $\mathcal{U}_{[a,b]}$ (until between $a$ and $b$), $\Diamond_{[a,b]}$ (eventually between $a$ and $b$) and $\Box_{[a,b]}$ (always between $a$ and $b$) that operate on a finite set of numerical predicates over the underlying signals. In discrete time setting, we assume that interval bounds $a$ and $b$ are nonnegative integers, $b > a$.

\begin{remark}
STL was originally developed for monitoring dense-time (continuous-time) signals. The reason for restriction of control synthesis to discrete time systems is translating bounded time model checking to mixed integer constraints by using a finite number of integers, a method that has been developed in \cite{raman}. It should be noted that while discrete time approximations may be used to resemble a continuous time system, validity of a STL formula in discretized time, in general, does not indicate validity of the formula in the original continuous system. 
\end{remark}

Basically, an STL specification for a control system is intended to require certain behavior from the system. The predicates determine thresholds for functions of state and control values. For example, $(x_1+x_2 \ge 1)$, is a  predicate over state values $x_1$ and $x_2$. 
In this paper, each predicate, $\mu_i$, is written in the form of:
\begin{equation}
\label{eq:predicates}
\mu_i:=\left(y_i[t] \geq 0 \right)~~~~~ i=1,\cdots,p,
\end{equation}
where the set of outputs $y=(y_1,\cdots, y_p)^T \in \mathbb{R}^p$ are required to be linear functions over the state and control in the form of \eqref{eq:secondary_signals}, and $p$ is the number of predicates. Matrices $C,D$ and vector $e$ are appropriately defined with respect to the specification predicates. 
For instance, $(x_1+x_2 \ge 1)$ corresponds to output $y=x_1+x_2-1$ where the predicate appearing in the STL formula is $(y\ge0)$. In this case, $C=(1~1), D=0$ and $e=-1$. 

Following the terminology from \cite{donze}, we refer to the state and control as \emph{primary} signals and to the scalars $y_i$ as  \emph{secondary} signals. Throughout the terminology used in this paper, the specification, which naturally is over the state and control, is written over the secondary signals. The secondary signals notion provides a consistent and convenient format for analyzing the STL constraints later in the paper. 

The validity of an STL formula is a function of the secondary signals. STL quantitive semantics defines a function called \emph{robustness} that associates a scalar for the quality of satisfaction. The robustness function is recursively defined as follows \cite{donze}:
\begin{equation}
\label{eq:semantics}
\begin{array}{l}
\rho_y^\mu[t]=y[t],
\\
\rho_y^{\neg \varphi} [t]=-\rho_y^\varphi[t], \\
\rho_y^{\varphi \vee \psi}[t]=\max (\rho_y^\varphi[t],\rho_y^\psi[t] ),
\\
\rho_y^{\varphi \wedge \psi} [t]=\min (\rho_y^\varphi[t],\rho_y^\psi[t] ),
\\
\rho_y^{\Diamond_{[a,b]} \varphi}[t]=\underset{t^\prime \in [t+a,t+b]} \max \rho_y^\varphi[t^\prime],
 \\
\rho_y^{\Box_{[a,b]} \varphi}[t]=\underset{t^\prime \in [t+a,t+b]} \min \rho_y^\varphi[t^\prime],
 \\
\rho_y^{\varphi~\mathcal{U}_{[a,b]} \psi}[t]=  \underset{t^\prime \in [t+a,t+b]} \max \left( \min (\rho_y^\varphi[t^\prime], \underset{t^{\prime\prime} \in [t,t^\prime]} \min \rho_y^\psi[t^{\prime\prime}])\right ),
\end{array}
\end{equation}
where $\mu=(y \geq 0)$ is a predicate over scalar $y$ and $\varphi$ and $\psi$ are STL formulas. Basically, $\rho^\varphi_y[t] > 0$ indicates that the formula $\varphi$ is satisfied by secondary signals starting at time $t$ whereas negative robustness indicates violation.
\begin{remark}Zero robustness, which has measure zero in the continuous domain, does not indicate satisfaction nor violation. In this paper, however, algorithms are developed such that zero robustness is considered as satisfaction.
 \end{remark}

\begin{example}
\label{example:system}
Consider the STL formula $\varphi=\Box_{[0,2]} \mu_1 \wedge \Diamond_{[0,3]} \mu_2$ where $\mu_i=(y_i \geq 0), i=1,2$. In words, this formula requires that $\mu_1$ is always true between time zero and two, \emph{and}, $\mu_2$ is eventually true between time zero and three. By applying the quantitative semantics in \eqref{eq:semantics}, the robustness function becomes: 
\begin{equation*}
\begin{split}
\rho_y^\varphi[t]=\min & (\min (y_1[t],y_1[t+1],y_1[t+2]), \\ & \max (y_2[t],y_2[t+1],y_2[t+2],y_2[t+3])).
\end{split}
\end{equation*}
Suppose the values of $y_i[t],t=0, \cdots, 5$ are given in Table \ref{table:signals}:
\begin{table}[H]
\begin{center}
   \caption{$y_1[t]$ and $y_2[t]$ for Example \ref{example:system}}
          \label{table:signals}

  \begin{tabular}{ | l | *{6}{c}  | }
    \hline
    t & 0 & 1 & 2 & 3 & 4 & 5 \\ \hline
    $y_1[t]$ & -0.5 & 1.5 & 1 & 1 & 0.8 & -0.5 \\ \hline
    $y_2[t]$ & 3 & 2 & 0.5 & -1 & -1.5 & -1 \\
    \hline
  \end{tabular}
   \end{center}

\end{table}
The robustness values are $\rho_y^\varphi[0]=-0.5$, $\rho_y^\varphi[1]=1$, $\rho_y^\varphi[2]=0.5$. Note that the values of $\rho_y^\varphi[t],t \geq 3$ are not computable, in general, without knowing the values of $y_i[t], t \geq 6$.

\endproof

\end{example}

Not that the robustness function depends on the future values of secondary signals. In this paper, we are interested in finding controls in a receding horizon manner that evolve the system such that an STL formula $\varphi$ is \emph{globally satisfied}, in the sense that the robustness function is always nonnegative, $\rho_y^\varphi[t] \geq 0 ~ \forall t \in \mathbb{Z}_{\geq 0}$. Global satisfaction can also be viewed as equivalent to satisfying ${\Box}_{[0,\infty]} \varphi$. 

Now we are ready to informally state the problem. Our primary aim is to develop a receding horizon controller that is \emph{correct}, \emph{robust} and \emph{optimal}. 
Note that, like other finite horizon controllers, suboptimal solutions are sought since the control synthesis algorithm depends on the size of the horizon. Furthermore, an issue in constrained controllers is infeasibility that often arises in highly disturbed or under-actuated systems. In this paper, when encountering infeasibility, we soften the STL predicates using slack variables. For example, $(-2 \geq 0)$ is softened as $(-2+\zeta \geq 0)$, which is satisfied by $\zeta \ge 2$.  Details on optimality criterion and predicate softening are given in Section \ref{sec:optimization}. 
\begin{problem}
\label{problem:problem}
Given a discrete-time system \eqref{eq:dynamics},\eqref{eq:secondary_signals}, an STL formula $\varphi$ over the predicates in the form of \eqref{eq:predicates}, a stage cost function $J(x[t],u[t])$, find a receding horizon controller that is:
\begin{enumerate}
\item Correct and Robust: The STL specification $\varphi$ is globally satisfied for all realizations of disturbances,
\item Optimal: cost $J(x[t],u[t])$ cumulated over the receding horizon is optimized, and
\item Minimally Violating: in case the global (finite horizon) satisfaction of $\varphi$ is infeasible, find controls such that the constraints are minimally violated.
\end{enumerate}
\end{problem}
After explaining necessary details in Section \ref{sec:robust}, the problem is formulated as an optimization problem in Section \ref{sec:optimization}.

\section{STL Receding Horizon Control}
\label{sec:rhc}


In this section, we describe the receding horizon mechanism of STL control. 
%
As mentioned earlier, the robustness function depends on the future values of the secondary signals. The time bounds of the temporal operators determine the future times of the secondary signals that are necessary to compute robustness.
\begin{define}
The \emph{horizon} length of STL formula $\varphi$, denoted by $h^\varphi$, is defined as the largest $\tau$ such that that robustness $\rho_y^\varphi[t]$ depends on $y[t+\tau]$. 
 \end{define}
The formula horizon can be recursively computed as \cite{dokhanchi}:
\begin{equation}
\label{eq:horizon}
\begin{array}{l}
h^\mu=0, \\
h^{\neg \varphi}=h^\varphi, \\
h^{\Diamond_{[a,b]}\varphi}=h^{\Box_{[a,b]}\varphi}=b+h^{\varphi},\\
h^{\varphi \wedge \psi}=h^{\varphi \vee \psi}=\max(h^{\varphi},h^{\psi}), \\
h^{\varphi \mathcal{U}_{[a,b]} \psi}=b+\max(h^{\varphi},h^{\psi}), \\ 
\end{array}
\end{equation}
where $\mu$ is a numerical predicate and $\varphi,\psi$ are STL formulas. In discrete time, robustness at a given time is a function of $h^\varphi$ steps of secondary signals in future, i.e. $\rho_y^\varphi[t]$ is a function of $y[t],y[t+1],\cdots,y[t+h^\varphi]$. 
For instance, the horizon length of the specification in Example \ref{example:system} is 3. It should be noted that the \emph{horizon length} of an STL formula, which follows from the terminology used in \cite{dokhanchi}, should not be confused with the \emph{horizon length} of the control sequence that is used in receding horizon control. 

At a given time $t$, the latest robustness that is computable is $\rho_y^\varphi[t-h^\varphi-1]$ given the history of values of secondary signals $y^{his}[t]=\{y[t-h^\varphi-1], \cdots,y[t-1] \}$. The robustness values starting from $\rho_y^\varphi[t-h^\varphi]$ depend on the values of secondary signals starting from time $t$ that are determined by the evolution of the system. We use the system model to predict robustness values and maintain satisfaction, i.e. nonnegative robustness, of the STL specification. Let $\rho_y^\varphi[t+h_p]$ be the most distant in future robustness value that is computed using the  predicted values of $y[t],\cdots,y[t+h_p+h^\varphi]$, where $h_p$ is  the \emph{prediction horizon} that is a user chosen integer. Consecutively, at time $t$, we enforce the following constraints to maintain the global STL satisfaction:
\begin{equation}
\label{eq:constraints}
\left \{
\begin{array}{cc}
\rho_y^\varphi[t-h^\varphi] & \geq 0, \\
\rho_y^\varphi[t-h^\varphi+1] & \geq 0, \\
\vdots & \\ 
\rho_y^\varphi[t+h_p] & \geq 0.
\end{array}
\right.
\end{equation}
Note that when starting the control software, while $t\le h^\varphi$, the set of constraints begin from $\rho_y^\varphi[0]$.
The horizon of predictions for secondary signals is $H=h^\varphi+h_p$ and at each time, the control sequence that is searched for is:
\begin{equation}
\label{eq:control}
u^H[t]=\{ u^t[t],u^{t}[t+1],\cdots,u^t[t+H] \}.
\end{equation}
Note that only the current time control is applied to the system and according to the new measurements, a new control sequence is found at next time step.  
\begin{remark}
In case of $D=0$ in \eqref{eq:secondary_signals}, i.e. secondary signals only functions of state, we can reduce the number of constraints and variables for faster computation. At time $t$, the value of $y[t]$ is already determined hence $\rho^\varphi_y[t-h^\varphi]$ is fixed. Therefore, the set of constraints can be written as $\rho^\varphi[t-h^\varphi+1] \geq 0, \cdots, \rho^\varphi[t+h_p]  \geq 0$ and since $y[t+H]$ is not dependent on $u^t[t+H]$, the finite horizon control sequence \eqref{eq:control} is $u^H[t]=\{u^t[t],u^t[t+1],\cdots,u^t[t+H-1]\}$. \end{remark}

The following theorem establishes closed-loop soundness of the receding horizon controller. 
\begin{theorem}
If ${u}^H[t]$ is found for all $t\in \mathbb{Z}_{\geq 0}$ such that \eqref{eq:constraints} is satisfied, then the evolution of the system from applying the closed loop control sequence
\begin{equation*}
u^0[0],u^1[1], \cdots, u^t[t],\cdots
\end{equation*}
globally satisfies $\varphi$. 
\end{theorem}
\begin{proof}
At time $t$, constraint $\rho_y^\varphi[t-h_p] \geq 0$ is satisfied and $\rho_y^\varphi[t-h_p]$ is fixed since it is not dependent on the further evolution of the system. 
By applying $u^t[t]$, $\rho_y^\varphi[t-h_p]+1 \geq 0$ is guaranteed since $u^t[t]$ was found such that all constraints in  \eqref{eq:constraints} were satisfied. By induction,  $\rho_y^\varphi[t] \geq 0 ~  \forall t \in  \mathbb{Z}_{\geq 0}$ and the specification $\varphi$ is hence globally satisfied.
\end{proof}
The theorem above does not state that $u^H[t]$ satisfying \eqref{eq:constraints} always exist. Intuitively, larger values of prediction $h_p$ may result in better performance due to longer horizon feasibility. However, in uncertain systems, larger horizon robust controllers are usually excessively conservative, and even infeasible since worst case predictions are made for a longer time. We address infeasibility issues in Sec. \ref{sec:optimization}.

\begin{remark}
The closed-loop synthesis approach of this paper differs from the approach in \cite{raman2015reactive}. In the mentioned work, at each time step the control sequence $u^{2h^\varphi}[t]$ is found such that the set of constraints $\rho_y^\varphi[t] \geq 0, \rho_y^\varphi[t+1] \geq 0,\cdots \rho_y^\varphi[t+h^\varphi] \geq 0$ are satisfied and controls $u^t[t],\cdots,u^t[t+h^\varphi-1]$ are fixed by the values found in the previous iterations. Therefore, the synthesis algorithm is based on maintaining STL satisfaction in future whereas the past robustness values are fixed since the control values $u^t[t],\cdots,u^t[t+H-1]$ are not updated anymore, even though they are not yet applied to the system. The control strategy also involves a transient phase where control inputs for initial steps are computed by a slightly different algorithm.
The receding horizon scheme in this paper, on the other hand, is based on both past and future satisfaction maintenance where by measuring the current state and storing the history of the system, the control sequence starting at the current time is updated at each time step. Therefore, our approach is more appropriate for online implementation since in the presence of uncertainties, given online measurements, the controller is able to find solutions that the control sequence stitching approach in \cite{raman2015reactive} may not. 
\end{remark}

\section{Robust STL Satisfaction}
\label{sec:robust}

In this section, we explain our approach to construct a robust MPC mechanism to find controls such that the system evolution satisfies the STL constraints for all possible realizations of disturbances. First, we introduce the notion of positive normal form STL which is important to characterize the behavior of STL constrains with respect to the changes in the secondary signals values. Next we provide the technical details on our approach to STL robust control. 

{\bf Notation} For two same-length vectors $a$ and $b$, inequalities such as $a\ge b$ are interpreted element-wise. The notation $\underline{1}$ stands for appropriately sized vector of all ones. The notation $\otimes$ denotes the Kronecker product.

\subsection{Positive Normal Form STL}

The robustness function constructed from the quantitative semantics in \eqref{eq:semantics} is a $\min/\max$ function that is piecewise linear and, in general, non-convex. The only scaling operation that appears in the quantitative semantics is multiplication by $-1$ from applying the semantics for negation operator.

\begin{define}
An STL formula $\varphi$ is in \emph{positive normal form} if its robustness is a non-decreasing function with respect to the secondary signals, i.e. if $y_1$ and $y_2$ are two secondary signals such that $y_1[\tau] \geq y_2[\tau] ~\forall \tau \in [t,t+h^\varphi]$, then $\rho_{y_1}^\varphi[t] \geq  \rho_{y_2}^\varphi[t]$. 
\end{define}
By writing the STL specification in positive normal form, the constraints become monotonic with respect to the values of secondary signals, which enables us to characterize the bounds of propagation of uncertainties into STL constraints. Furthermore, in Section \ref{sec:optimization}, we exploit the positive normal form structure to soften the STL constraints by the means of addition of slack variables to the secondary signals.   

\begin{proposition}
An STL specification that does not contain negation, i.e. only consisting of $\wedge,\vee, \mathcal{U}_\mathcal{I},  \Diamond_\mathcal{I}, \Box_\mathcal{I}$ where $\mathcal{I}$ is a time interval, is in positive normal form. 
\end{proposition}
\begin{proof}(sketch)
By excluding the negation operator, the robustness function consists only of $\min$ and $\max$ operators over the secondary signal values without sign change. Therefore, an increase in secondary signal values will either increase robustness, or does not alter robustness. 
\end{proof}

\begin{proposition}
Every STL formula can be transformed into positive normal form. 
\end{proposition}
\begin{proof}
We explain how to recursively transform a general STL formula to positive normal form using negation propagation into numerical predicates and modifying the sign of secondary signals:
 \begin{enumerate}
\item Negation recursively  propagates into numerical predicates:
\begin{equation*}
\begin{array}{l} 
\neg (\varphi_1 \wedge \varphi_2) = \neg \varphi_1 \vee \neg\varphi_2,
\\
\neg (\varphi_1 \vee \varphi_2) = \neg \varphi_1 \wedge \neg\varphi_2, \\
\neg (\varphi_1 \mathcal{U}_{[a,b]} \varphi_2) = \neg\varphi_1 \mathcal{R}_{[a,b]} \neg\varphi_2, \\ 
\neg \Diamond_{[a,b]} \varphi = \Box_{[a,b]} \neg\varphi, \\
\neg \Box_{[a,b]} \varphi= \Diamond_{[a,b]} \neg\varphi,
\end{array}
\end{equation*}
where $\mathcal{R}$ is the \emph{release} operator defined as $\varphi_1 \mathcal{R}_{[a,b]} \varphi_2=\Box_{[a,b]} \varphi_2 \vee (\varphi_2 \mathcal{U}_{[a,b]} (\varphi_1\wedge \varphi_2))$.
\item Negation of predicates: If $\neg \mu_i$ appears in the STL formula where $\mu_i=(y_i\ge 0)$ that $y_i=c_i x+d_i u+e_i$ is a secondary signal: 
\begin{itemize}
\item If only $\neg \mu_i$ appears in the STL formula: Redefine $y_i=-c_i x-d_i u-e_i$, thus $\neg$ is removed. 
\item If both $\mu_i$ and $\neg \mu_i$ appear: Introduce new secondary signal $y_{j}=-c_i x-d_i u - e_i$, thus $\neg \mu_i$ is replaced by the new predicate $\mu_j=(y_j \geq 0)$.
\end{itemize} 
\end{enumerate}
\end{proof}
Note that, in the worst case, the number of secondary signals is doubled during the construction of the positive normal form STL formula. As explained in Section \ref{sec:optimization}, the computational complexity of control synthesis is exponential with respect to the number of secondary signals. 

\subsection{Robust Prediction System}

At time $t$, given the control sequence $u^H[t]=(u^t[t]^T,\cdots,u^t[t+H]^T)^T$ and the uncertain input sequence $w^H[t]=(w^t[t]^T,\cdots,w^t[t+H-1]^T)^T$, the prediction for secondary signals $y^H[t]=(y^t[t]^T, \cdots,y^t[t+H]^T )^T$ is:
\footnote{With a slight abuse of notation, $w^H[t]$,$u^H[t]$ and $y^H[t]$ are interchangeably used for the described sequences and their corresponding representations as column vectors.}
\begin{equation}
\label{eq:matrix}
y^H[t]
=
\Phi_0^H x[t]
+
\Phi_1^H u^H[t]
+ 
\Phi_2^H w^H[t]
+
\underbar{1} \otimes e,
\end{equation}
where the flow matrices $\Phi_0,\Phi_1,\Phi_2$ are given by:
\begin{equation*}
\label{eq:phi}
\Phi_0^H=
\left(
\begin{array}{c}
C \\ CA \\ \vdots \\ CA^H
\end{array}
\right ),
\end{equation*}
\begin{equation*}
\Phi_1^H
=
\left(
\begin{array}{ccccc}
D & 0 && \cdots & 0 \\
CB & D && \cdots & 0 \\
CAB & CB && \cdots &0 \\
\vdots & \vdots &  \ddots &\ddots& \vdots \\
CA^{H-1}B & CA^{H-2}B & \cdots & CB & D 
\end{array}
\right),
\end{equation*}
\begin{equation*}
\Phi_2^H= 
\left(
\begin{array}{cccc}
0 & 0 & \cdots & 0 \\
C & 0 & \cdots & \vdots \\
CA & C & \cdots &\vdots \\
\vdots & \vdots &  \ddots & \vdots \\
CA^{H-1} & CA^{H-2} & \cdots & C 
\end{array}
\right).
\end{equation*}

Since the uncertain inputs belong to the polytope set $\mathcal{W}$, the admissible set of $y^H[t]$ is also a polytope in the finite horizon secondary signals space:
\begin{equation}
\mathcal{Y}_{u^H[t]}=\left \{ y^H[t] \middle| w^t[\tau] \in \mathcal{W}, \tau=t,\cdots,t+H-1 \right\},
\end{equation}
which consists of \emph{uncertainty image} set, $\left\{ \Phi_2 w^H[t] \middle| w^t[\tau] \in \mathcal{W} \right\} $, which can be computed beforehand, plus an affine map of the control sequence $\Phi_0^H x[t]+\Phi_1^Hu^H[t]+\underline{1}\otimes e$. The finite horizon robust prediction problem is finding controls $u^H[t]$ such that the set of constraints \eqref{eq:constraints} is satisfied for all points in $\mathcal{Y}_{u^H[t]}$. 
\begin{figure}[t]
\centering
\begin{tikzpicture}[xscale=0.55,yscale=0.8]
\draw[fill=cyan]  (4,2) -- (3,3) -- (5,3) -- (6,2) -- cycle;
\node[] at (3.5,4.4) { $y[t+2]$};
\node[] at (7.6,2.5) { $y[t+1]$};
\draw[->] (2,2.5) -- (6.4,2.5);
\draw[->] (3.8,1.5) -- (3.8,4);
\draw (4.0,1.2) node[] {a)};

\end{tikzpicture}
\begin{tikzpicture}[xscale=0.55,yscale=0.8]
\draw[dashed]  (3,2) -- (6,2) -- (6,3) --  (3,3) -- cycle;
\draw[fill=cyan]  (4,2) -- (3,3) -- (5,3) -- (6,2) -- cycle;
\node[] at (3.5,4.4) { $y[t+2]$};
\node[] at (7.6,2.5) { $y[t+1]$};
\draw[->] (2,2.5) -- (6.4,2.5);
\draw[->] (3.8,1.5) -- (3.8,4);
\draw (3,2) node[] {\textbullet};
\draw (4.0,1.2) node[] {b)};

\end{tikzpicture}
\caption{a) An example representation of the admissible set of $y^H[t]$ in the finite horizon secondary signals space. b) The lower left corner of the axis-aligned minimum bounding box of the set. }
\label{fig:shaded}
\end{figure}
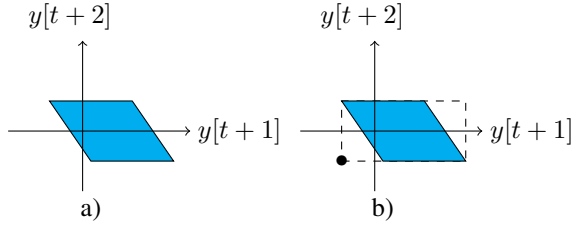
Since the robustness function is non-convex, its extreme values do not necessarily lie on the vertices of $\mathcal{Y}_{u^H[t]}$. For example,
consider a simple robustness function $\rho_y^\varphi=\max(y[t+1],y[t+2])$ and assume that for some $u^H[t]$, the set $\mathcal{Y}_{u^H[t]}$ is the shaded parallelogram illustrated in Fig. \ref{fig:shaded} a). It is observed that even though all the vertices of the parallelogram are in the positive robustness region, i.e. the first, second and fourth quadrants in Fig. \ref{fig:shaded} a), a small section lies in the third quadrant which corresponds to negative robustness. 

In order to maintain robust satisfaction, we enforce the constraints \eqref{eq:constraints} at a single point that is the lower left corner of the axis-aligned minimum bounding box of the uncertainty image set (See Fig. \ref{fig:shaded} b) for an illustration).
As stated in Theorem \ref{theorem_robust} later in the paper, with STL formula being in positive normal form, the robustness is guaranteed to be greater or equal to the robustness at the lower left corner of the box since the secondary signals of the image are greater element-wise. Therefore, the robust \emph{prediction} system is constructed based on the mentioned lower left corner. Note that this approach is, in general, conservative yet computationally manageable as the set of constraints \eqref{eq:constraints} are imposed for a single point.

\begin{define}
The lower left corner of the axis-aligned minimum bounding box of a bounded set $\mathcal{S} \subset \mathbb{R}^n$ is denoted by $\Omega(\mathcal{S})$, where the $i$'th element is given by:
\begin{equation}
\Omega_i(\mathcal{S})=\underset{s \in \mathcal{S} }\inf \quad s_i, ~i=1,\cdots,n.
\end{equation}
\end{define}

A polytope set $\mathcal{P} \subset \mathbb{R}^n$ can be represented by the convex hull of its vertices. For a given polytope $\mathcal{P}$, we define the matrix $P$ whose columns contain its vertices. 
%
%
%
\begin{define}
The function $ \omega: \mathbb{R}^{n\times m} \rightarrow  \mathbb{R}^n$ maps a $n\times m$ matrix to a $n$-dimensional vector where the $i$'th element is:
\begin{equation}
\omega_i(P)=\underset{j} \min \quad P_{ij}.
\end{equation}
 \end{define}
 
In words, the $i$'th element of the vector is the minimum value in the $i$'th row of the matrix. 
\begin{lemma}
\label{lemma:omega}
Let $\mathcal{P}$ be a polytope and matrix $P$ whose columns contain its vertices. Then:
\begin{equation}
\label{eq:omega_equal}
\Omega(\mathcal{P})=\omega(P).
\end{equation}
\end{lemma}
\begin{proof}
The value of $i$'th element is the solution to the following linear program: 
\begin{equation*}
\begin{array}{cl}
\Omega_i(\mathcal{P})=\min & \sum\limits_{j} P_{ij} \lambda_{jk}, \\
s.t.&0\le\lambda_{jk} \le1,\\
&  \sum\limits_{k} \lambda_{jk} =1.
\end{array}
\end{equation*}
It is straightforward to verify that the solution is:
\begin{equation*}
\lambda_{jk}= \left \{ 
\begin{array}{ll}
1 & j= argmin~ P_{ij}, \\
0 & otherwise.
\end{array}
  \right. 
\end{equation*}
Therefore, $\Omega_i(\mathcal{P})=\omega_i(P)$ holds element-wise. 
\end{proof}


\begin{lemma}
\label{lemma_secondary}
If the open-loop control sequence $u^H[t]$ is applied to the system \eqref{eq:dynamics},\eqref{eq:secondary_signals}, the following relation holds:
\begin{equation}
y^{\omega,H}[t] \le y^H[t],
\end{equation}
where
\begin{equation}
\label{eq:predicted}
y^{\omega,H}[t]= \Phi_0 x[t]+\Phi_1 u^H[t]+ \omega\left(\Phi_2 (\underline{1}\otimes W)\right)+\underline{1}\otimes e,
\end{equation}
where $W$ is the matrix whose columns are given by the polytope $\mathcal{W}$'s vertices. 
\end{lemma}
\begin{proof}
The proof follows from the definition of $\Omega$ function where
\begin{equation}
y^{\omega,H}[t]= \Omega \left (\left\{ \mathcal{Y}_{u^H[t]} \right\} \right) \le y^H[t].
\end{equation}
By applying Lemma \ref{lemma:omega} to \eqref{eq:matrix} we arrive at \eqref{eq:predicted}.
\end{proof}

Note that $\omega\left(\Phi_2 (\underline{1}\otimes W)\right)$ is computed prior to starting the control synthesis optimization problem. 

\begin{theorem}
\label{theorem_robust}
Given a linear system \eqref{eq:dynamics}, \eqref{eq:secondary_signals}, STL specification $\varphi$ in positive normal form and secondary signals history $y^{his}[t]$, for any control sequence $u^H[t]$ the following relations hold:
\begin{equation}
\label{eq:robust_predict}
\begin{array}{lcl}
\rho_y^\varphi[t-h^\varphi] & \geq & \rho_{y^{pre}}^\varphi[t-h^\varphi], \\
\rho_y^\varphi[t-h^\varphi+1] & \geq & \rho_{y^{pre}}^\varphi[t-h^\varphi+1], \\
& \vdots & \\ 
\rho_y^\varphi[t+h_p] & \geq & \rho_{y^{pre}}^\varphi[t+h_p],
\end{array}
\end{equation}
where $y^{pre}$ is the prediction secondary signal that is composed from the stored values of $y^{his}$ and the robust prediction values from \eqref{eq:predicted}.
\end{theorem}
\begin{proof}
The proof follows immediately from Lemma \ref{lemma_secondary} and the definition of positive normal form STL. 
\end{proof}
\begin{corollary}
If the control sequence $u^H[t]$ is found such that the set of constraints
\begin{equation}
\label{eq:constraints_robust}
\left \{
\begin{array}{cc}
\rho_{y^{pre}}^\varphi[t-h^\varphi] & \geq 0, \\
\rho_{y^{pre}}^\varphi[t-h^\varphi+1] & \geq 0, \\
\vdots & \\ 
\rho_{y^{pre}}^\varphi[t+h_p] & \geq 0,
\end{array}
\right.
\end{equation}
are satisfied, then the open-loop system response of $u^H[t]$ satisfies the set of original constraints \eqref{eq:constraints}. 
\end{corollary}

\begin{remark}
The methodology of this paper can be easily extended to linear time variant (LTV) systems. The required modification is generalizing the flow matrices in \eqref{eq:matrix} for time dependent matrices. In this case, the necessary assumption is that the time dependencies of the system matrices are known. 
\end{remark}

\begin{remark}
We have not assumed any restriction on the plant matrix $A$. In principle, an unstable $A$ results in large entries in matrices in \eqref{eq:matrix} that causes control decisions to be very conservative and may even cause infeasibility in longer horizon predictions. A well known technique in MPC literature is designing a control law in the form of $u[t]=Kx[t]+v[t]$, where $K$ is a fixed state feedback gain. The closed loop matrix $A^{cl}=A+BK$ can be designed to possess some important properties such as stability and nil-potency (if the pair $(A,B)$ is controllable). We have not investigated this approach since STL constraints, in general, are different from stability. We also remark that an analogous investigation of robust invariant sets \cite{bemporad1999robust} for STL MPC is an open problem.
\end{remark}

\section{Optimization Based Control}
\label{sec:optimization}

In the previous section, we explained our approach to the first objective of Problem \ref{problem:problem}. In this section, after explaining our approach to the formalization of the second and third objectives, i.e. optimality and minimality of violations, we formulate Problem \ref{problem:problem} as an optimization problem. Finally we explain how to express the optimization problem as a mixed integer programming (MIP) problem that is solvable using standard solvers.

\subsection{Optimization Problem}
A performance criterion is required for selecting a control sequence from the robust open-loop control candidates $u^H[t]$. In principle, there exist two different approaches to define a cost criteria for a nondeterministic system. First, one can optimize the cost using predictions from the nominal system, where the disturbances are assumed to take a known \emph{nominal} value. A more complicated alternative is optimizing the worst case cost that is admissible by the disturbance realizations. In this paper, we choose the nominal system cost since it is found to perform better in many classical MPC problems \cite{bemporad1999robust}. Furthermore, worst case cost approaches lead to optimization problems that are computationally more expensive. It should be noted that if the cost is only dependent on controls, the two approaches are identical.

We define $\hat{w}$ as the \emph{nominal} disturbance, that may be given or may be chosen by some means such as finding the centroid of the polytope $\mathcal{W}$. Given the current state $x[t]$, the nominal system prediction is given by 
\begin{equation}
\label{eq:nominal}
\begin{split}
\hat{x}[\tau+1]=& A\hat{x}[\tau]+Bu[\tau]+\hat{w}, \\ & t \le \tau \le t+H-1, \\ & \hat{x}[t]=x[t].
\end{split}
\end{equation}
At time $t$, the finite horizon cost function is:
\begin{equation}
\sum \limits_{\tau=t}^{t+H} J(\hat{x}[\tau],u[\tau]).
\end{equation}
Effectively, within all control sequences that robustly satisfy STL constraints, we choose the one with the least finite horizon nominal evolution cost. 

On the other hand, in systems with large disturbances or limited control actuation, it is possible that the STL constraints may be inevitably violated. 
If encountered with infeasibility, instead of terminating the control software, we find minimally violating solutions. With STL formula $\varphi$ in positive normal form, a counterfeit increase in all the secondary signals values eventually restores the satisfaction of STL constraints. This method is similar to constraint softening method in \cite{kerrigan2000soft}. We introduce softened secondary signals values as:
\begin{equation}
\label{eq:soft}
y^{his}_{soft}=y^{his}+\underline{1} \zeta,
~~~~~
y^{\omega,H}_{soft}=y^{\omega,H}+\underline{1} \zeta,
\end{equation}
where $\zeta \geq 0$ is the softening slack variable. Note that both robust prediction values and history values are softened to recover feasibility of \eqref{eq:constraints}. The artificial secondary signal composed from $y^{his}_{soft}$ and $y^{\omega,H}_{soft}$ is denoted by $y^{soft}$. 
We desire that if the STL constraints are infeasible, $\zeta$, the amount of softening, is minimized without optimizing the cost function. Finally, Problem \ref{problem:problem} is formulated as the following optimization problem:
\begin{equation}
\label{eq:optimization}
\begin{array}{ccl}
u^H[t]= & argmin & \sum \limits_{\tau=t}^{t+H} J(\hat{x}[\tau],u[\tau]) + M \zeta \\
& s.t.&  \rho_{y^{soft}}^\varphi[t-h^\varphi] \geq0 \\
& & \vdots \\
& &  \rho_{y^{soft}}^\varphi[t+h_p] \geq 0, \\
& & \text {Eqn. \eqref{eq:predicted}, \eqref{eq:nominal}, \eqref{eq:soft}},  \\
& & \zeta \geq 0,
\end{array}
\end{equation}
where $M$ is a large penalizing positive number that unifies the separate optimization problems for cost optimality and violation minimality. In case the STL constraints are feasible, large $M$ enforces $\zeta=0$ and the cumulated cost is optimized. In case of infeasibility, effectively, $\zeta$ is minimized without optimization of the cumulated cost.

\begin{proposition}
The smallest $\zeta$ such that the constraints set in \eqref{eq:optimization} is feasible is:
\begin{equation}
\zeta_{min}= \max\{0,-\underset{\tau=t-h^\varphi,\cdots,t+h_p}\min \rho_{y^{pre}}^\varphi[\tau] \}.
\end{equation}
\end{proposition}
The detailed proof is not included as it basically follows from the monotonicity of the robustness function of a positive normal form STL formula. 

\begin{remark}
Eqn. \eqref{eq:soft} can be modified by using weights for softening different secondary signals. Multiple softening values for different secondary signals may also be used. In these cases, a practically efficient controller may require a tuning procedure to find the best softening strategy. 
\end{remark}

\begin{remark}
Removing the constraint $\zeta \geq 0$ results in a STL robustness maximization receding horizon policy. A negative $\zeta$ value can be viewed as \emph{constraint tightening}, i.e. how much constraints can be tightened while keeping feasibility. 
\end{remark}

\subsection{Mixed Integer Formulation}
STL constraints can be written as mixed integer constraints. One can encode the robustness function by representing $\max$ and $\min$ operations in the quantitative semantics, Eqn. \eqref{eq:semantics}, by a set of mixed integer constraints. This method typically introduces a large number of integer variables as the number of $\max/\min$ arguments may be large. An alternative approach, which is computationally more efficient, is binary encoding of quantitative semantics, which has been first introduced by the authors in \cite{raman}. In this section, we briefly explain this method.

The binary encoding is recursively executed. For a single predicate $\mu=(y \geq 0)$, a binary $z^\mu[t] \in \{0,1\}$ indicates whether the predicate at time $t$ is true, $z^\mu[t]=1$, or false, $z^\mu[t]=0$. The corresponding mixed integer constraints are:
\begin{equation}
\left \{ 
\begin{array}{ll}
y[t]-Kz^\mu[t] & \le 0, \\
y[t]+K(1-z^\mu[t]) & \geq 0,
\end{array}
\right.
\end{equation}
where $K$ is a sufficiently large positive number. Overall, $p\times(H+h^\varphi+1)$ number of binary variables is required to binary encode the secondary signals in \eqref{eq:optimization}. 
For encoding the STL formula, an additional number of variables are recursively defined as \cite{raman}:
\begin{itemize}
\item Conjunction: $\psi=\bigwedge_{i=1}^m \varphi_i~$:
\begin{equation}
\left\{ 
\begin{array}{c}
z^\psi[t] \le z^{\varphi_i} \\
z^\psi[t] \ge 1-m+\sum \limits_{i=1}^m z^{\varphi_i}[t]
\end{array}
\right.
\end{equation}
\item Disjunction: $\psi=\bigvee_{i=1}^m \varphi_i~$:
\begin{equation}
\left\{ 
\begin{array}{c}
z^\psi[t] \ge z^{\varphi_i}[t] \\
z^\psi[t] \le \sum \limits_{i=1}^m z^{\varphi_i}[t]
\end{array}
\right.
\end{equation}
\end{itemize}

\begin{itemize}
\item Eventually $\psi=\Diamond_{[a,b]} \varphi$
\begin{equation}
z^\psi[t] = \bigvee_{\tau=t+a}^{t+b} z^{\varphi_i}[\tau]
\end{equation}
\item Always $\psi=\Box_{[a,b]} \varphi$
\begin{equation}
z^\psi[t] = \bigwedge_{\tau=t+a}^{t+b} z^{\varphi_i}[\tau] 
\end{equation}
\item Until $\psi=\varphi_1\mathcal{U}_{[a,b]} \varphi_2$
\begin{equation}
z^\psi[t] = \bigvee_{\tau=t+a}^{t+b} (z^{\varphi_2}[\tau]  \wedge \bigwedge_{\tau^\prime=t}^{\tau} z^{\varphi_1}[\tau^\prime])
\end{equation}
\end{itemize}
Note that each $z^\psi \in [0,1]$ is not required to be declared as an integer since is automatically enforced to take values from $0$ or $1$. Finally, the set of constraints \eqref{eq:constraints} becomes binary encoded as:
\begin{equation}
\left \{
\begin{array}{cc}
z^\varphi[t-h^\varphi] & = 1, \\
z^\varphi[t-h^\varphi+1] & = 1, \\
\vdots & \\ 
z^\varphi[t+h_p] & = 1.
\end{array}
\right.
\end{equation}
Depending on the cost function $J$, the optimization problem is a mixed integer linear programming (MILP) (in case of linear $J$), a mixed integer quadratic programming (MIQP) (in case of quadratic $J$) or a mixed integer nonlinear programming (MINLP) (in case of nonlinear $J$). Mixed integer programs are exponentially expensive with respect to the number of integer variables, therefore the real time applications of STL MPC are restricted to small systems.

\section{Case Study}
\label{sec:case}
We consider a linear system in the form \eqref{eq:dynamics}, with 
\begin{equation*}
A=
\left( \begin{array}{cc} 1 & 0.5 \\ 0 & 0.8 \end{array} \right),
B=
\left( \begin{array}{c} 0 \\ 1 \end{array} \right).
\end{equation*}
The two dimensional state is $x[t]=(x_1[t],x_2[t])^T$ and control input is a scalar. This system represents a double integrator with energy loss, a type of system which is encountered in many engineering problems.
The disturbance $w[t]$ is bounded to the two dimensional box $\left \|w[t] \right \|_\infty \le w_0$, where $w_0$ is assigned multiple values as explained further. The stage cost function is $J(u[t])=\left | u[t] \right |$, which penalizes the control effort.

We desire a STL specification that enforces $x_1$ to oscillate between $2\le x_1\le 4$ and $-4\le x_1\le -2$, with each interval being visited at least once within any five consecutive time steps. The specification, written in positive normal form is:
\begin{equation*}
\begin{split}
\varphi=&\left(\Diamond_{[0,4]} ((y_1\ge 0) \wedge (y_2 \ge 0))\right) \wedge \\ & \left(\Diamond_{[0,4]} ((y_3 \ge 0) \wedge (y_4 \ge 0))\right),
\end{split}
\end{equation*}
for which the corresponding matrices from \eqref{eq:secondary_signals} are:
\begin{equation*}
C=\left( \begin{array}{cc} 1 & 0 \\ -1 & 0 \\ -1 & 0 \\ 1 & 0 \end{array} \right), D=0, e= \left( \begin{array}{c} -2 \\ 4 \\ -2 \\ 4 
\end{array} \right).
\end{equation*}
We chose $h_p=2$, which makes $H=h^\varphi+2=6$. The initial state values are $x_1[0]=x_2[0]=0$. The control admissible set is initially assigned as $\left | u \right | \le 20$.  We formulate the optimization problem given in \eqref{eq:optimization} as a MILP and find solutions using the MATLAB standard optimization toolbox MILP solver. The assigned value of $M$ in \eqref{eq:optimization} is $10^5$. We simulate the system for 30 time steps. The solution of each step takes less than 0.1s using a 2.8 GHz core i5 processor on an iMac computer. The uncertain values $w[t]$ are drawn randomly from a uniform distribution on $\mathcal{W}$. We investigated the following scenarios:

\begin{enumerate}

\item We observe that if the system was fully deterministic, $w_0=0$, the optimal-correct solution oscillates between $x_1=2$ and $x_1=4$ as illustrated in Fig. \ref{fig:case} a). The solution does not enter the mentioned regions as it is unnecessary and is associated with a larger control effort. The robustness function is always zero for this solution. 

\item We consider a disturbed system $w_0 = 0.2$. We observe that the nominal MPC, i.e. neglecting the presence of disturbances, fails to satisfy the specification. The trajectory is shown in Fig. \ref{fig:case} b) and the robustness function values in this case are occasionally negative.

\item We now implement the robust MPC introduced in this paper. It is observed that the robust solution enters the regions to conservatively maintain STL satisfaction (See Fig. \ref{fig:case} c)). The robustness function is always above zero for this case.

\item We broaden the disturbances set to  $\left \|w[t] \right \|_\infty \le 0.5$. The controller fails to find a robust solution to this scenario, see Fig. \ref{fig:case} d), thus constraint softening is required. This is particularly due to the long horizon considered where the worst case predictions cause infeasibility.  We observed that by decreasing the horizon length to $H=5$ and $H=4$, better solutions were found (results not shown). 

\item We now limit the admissible control set to $\left | u \right | \le 2$ with disturbances from the set $\left \|w[t] \right \|_\infty \le 0.2$. It is impossible to satisfy STL constraints with such a limited control set. The implementation result, shown in Fig. \ref{fig:case} e), does not satisfy the STL specification, but nevertheless, the observed trajectory is maximally oscillating between the two regions in order to minimally violate the specification. 
\end{enumerate}

\begin{figure}[t]
\begin{center}
\begin{tabular}{@{}c@{}@{}c@{}}

            \includegraphics[width=0.23\textwidth]{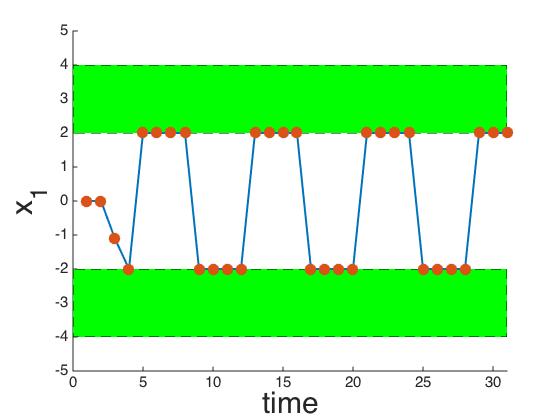} &  \includegraphics[width=0.23\textwidth]{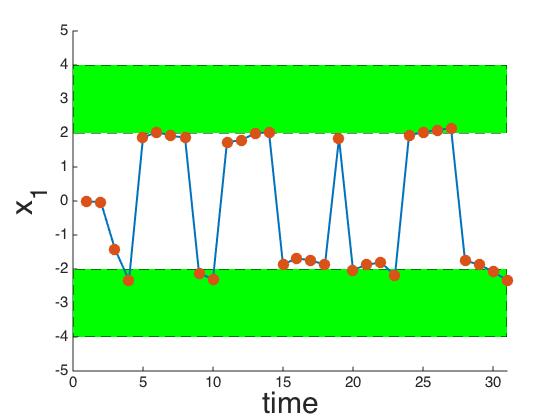}
                        \\
                                          a) $w[t]=0$ & b) $\left \|w[t] \right \|_\infty \le 0.2$ \\
                                                                                    MPC & nominal MPC
                        \\
                      \includegraphics[width=0.23\textwidth]{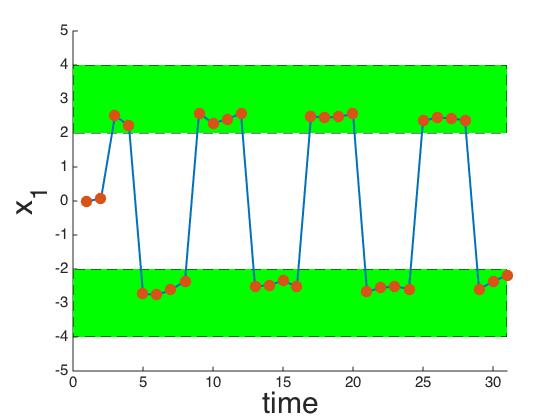} &    \includegraphics[width=0.23\textwidth]{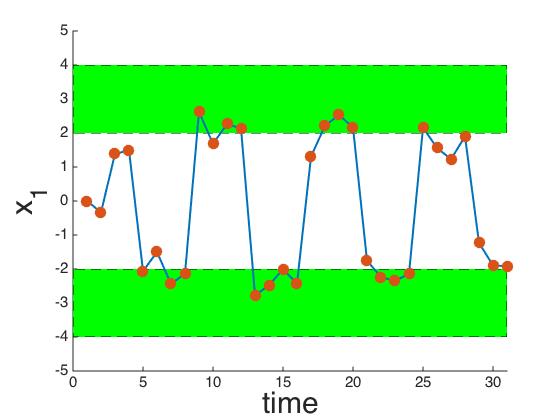} 
                        \\
                        c) $\left \|w[t] \right \|_\infty \le 0.2$  & d) $\left \|w[t] \right \|_\infty \le 0.5$
                        \\
                        robust MPC & soft constrained MPC
                        \\
                                              \includegraphics[width=0.23\textwidth]{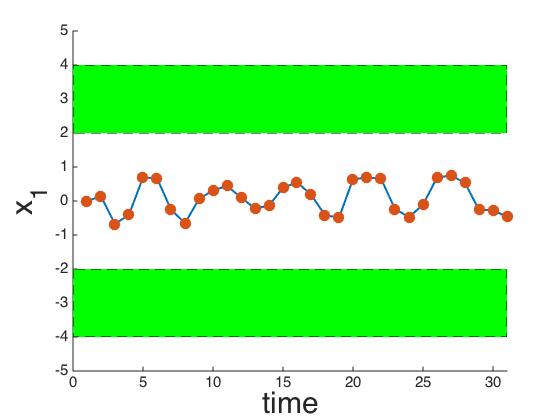} &    \includegraphics[width=0.23\textwidth]{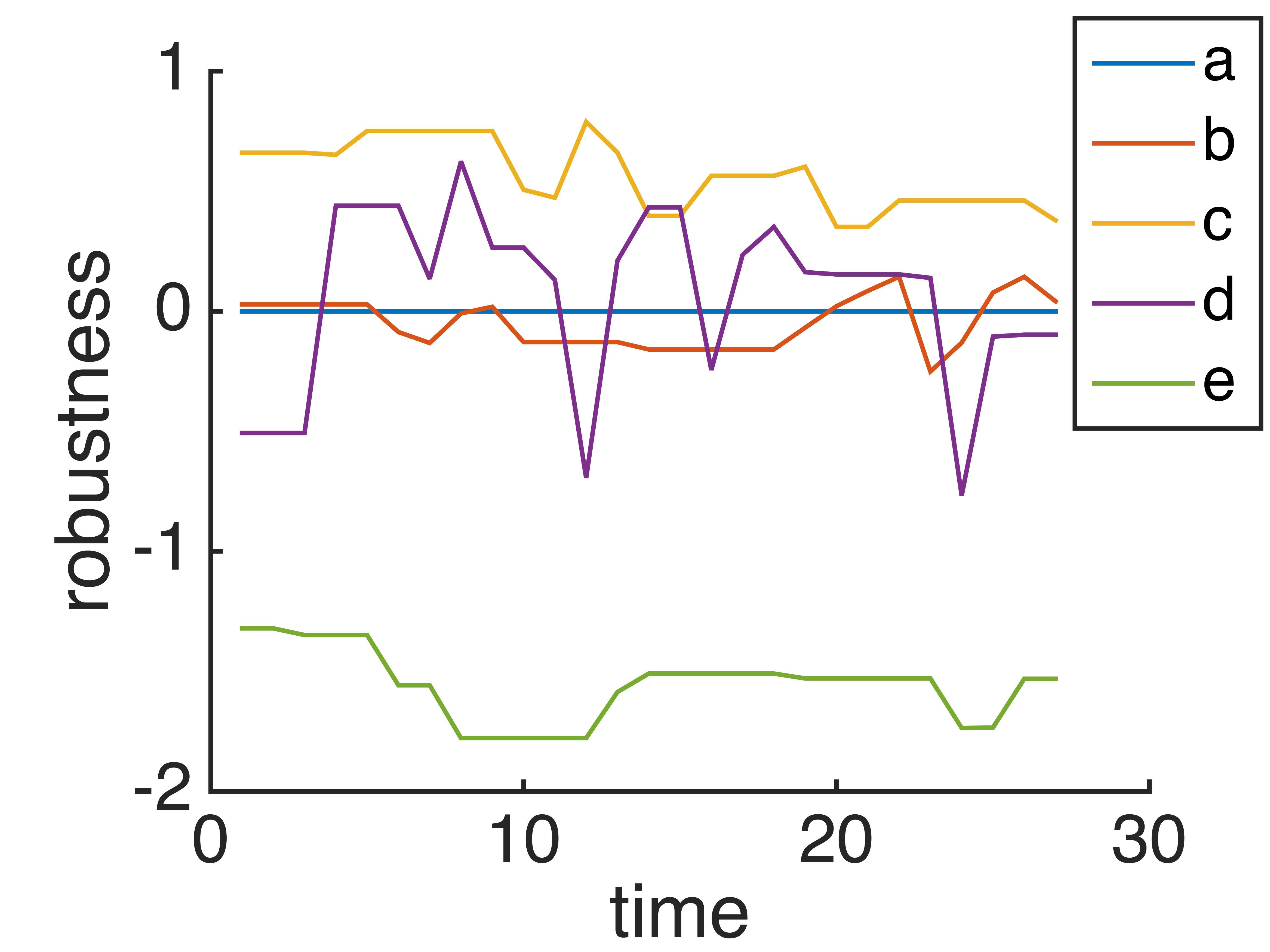}
                        \\
                        e) $\left |u[t] \right | \le 2$  & f) STL robustness function
                        \\
                        minimally violating MPC & values
                        
        \end{tabular}
\caption{Case study results}
\label{fig:case}
\end{center}
\end{figure}

The robustness values as a function of time are shown in Fig. \ref{fig:case} f) for the five different scenarios.

\section{Conclusion and Future Work}

In this paper, we focused on the connection between optimality and correctness for discrete time linear systems with additive bounded disturbances. Specifically, we proposed a model predictive control approach for the case when the correctness specifications are given as signal temporal logic formulas. We plan to extend these results to classes of discrete-time piecewise affine systems for which the monotonicity property stated in this paper holds, and apply them to controlling traffic networks. We are also working on extending this work to distributed model predictive control, where the additive disturbances of the component subsystems are used in assume-guarantee reasoning schemes. We believe that such techniques have the potential to impact temporal logic optimal control of large networks.

\bibliography{IEEEabrv,bib_robust}
\end{document}